\documentclass[11pt,reqno]{amsart}
\usepackage{amssymb}
\usepackage{amsthm}
\usepackage{amsfonts}
\usepackage{amsmath}
\usepackage{stmaryrd}
\usepackage{physics}
\usepackage{braket}
\usepackage{dsfont}
\usepackage{bbold}
\usepackage{graphicx}
\usepackage{relsize}
\usepackage[table,xcdraw]{xcolor}
\usepackage{enumerate}
\usepackage{hyperref}
\hypersetup{
    colorlinks=true,
    linkcolor=blue,
    citecolor=red,
    filecolor=magenta,      
    urlcolor=black,
}
\usepackage[margin=1in,top=0.8in,bottom=0.8in]{geometry}
\usepackage{enumitem}
\usepackage{mathtools}
\usepackage{graphbox}
\usepackage{comment}
\usepackage{float}
\usepackage[font=scriptsize]{caption}

\renewcommand{\epsilon}{\varepsilon}
\renewcommand{\phi}{\varphi}

\newcommand{\M}[1]{\mathbb{M}_{#1}(\mathbb{C})}
\newcommand{\C}[1]{\mathbb{C}^{#1}}
\newcommand{\N}{\mathbb{N}}
\newcommand{\cD}{\mathbb{D}}

\newcommand{\iden}{\mathbb{1}}
\newcommand{\id}{\mathrm{id}}
\newcommand{\U}[1]{\mathbb{U}(#1)}

\newtheorem{theorem}{Theorem}[]

\newtheorem*{definition*}{Definition}

\newtheorem{lemma}[theorem]{Lemma}

\newtheorem*{conjecture*}{Conjecture}

\theoremstyle{definition}

\definecolor{darkgreen}{rgb}{0,0.392,0}

\title{A short proof of the modified Kretschmann--Schlingemann--Werner conjecture}

\begin{document}

\begin{abstract}
Let $\Phi_1, \Phi_2 : \M{d}\to \M{n}$ be two quantum channels with respective Stinespring isometries $V_1, V_2 : \C{d}\to \C{n} \otimes \C{m}$ on \emph{any} common dilation space $\C{m}$. We prove that there exists a unitary $U$ on $\C{m}$ such that 
$\norm{V_1 - (\iden \otimes U) V_2}_{\infty} \leq \sqrt{2\norm{\Phi_1 - \Phi_2}_{\diamond}},$
thus resolving vom Ende's modification of the Kretschmann-Schlingemann-Werner conjecture in the affirmative.

\end{abstract}

\author{Satvik Singh}
\address{\small{\parbox{\linewidth}{Department of Mathematics, Technical University of Munich, Garching, Germany \\ 
Munich Center for Quantum Science and Technology (MCQST), Munich, Germany \vspace{0.1cm}}}}

\email{satvik.singh@tum.de}

\maketitle

\section{Introduction} 

For any completely positive and trace-preserving linear map (also called a \emph{quantum channel}) $\Phi: \M{d}\to \M{n}$, Stinespring's dilation theorem \cite{Stinespring1955} asserts that the action of $\Phi$ can be described by first isometrically embedding the input into a joint output-dilation space followed by tracing out the dilation (also called environment). More precisely, for any quantum channel $\Phi:\M{d}\to \M{n}$, there exists an isometry $V: \C{d}\to \C{n} \otimes \C{m}$ such that
\begin{equation}\label{eq:Stinespring}
\forall X\in \M{d} : \quad \Phi (X) = \Tr_{\C{m}} \bigl( V X  V^* \bigr) ,
\end{equation}
where $\Tr_{\C{m}}$ denotes the partial trace over $\C{m}$. For a fixed dilation space $\C{m}$, the Stinespring representation \eqref{eq:Stinespring} is unique up to unitary transformations on $\C{m}$. 

The continuity theorem of Kretschmann, Schlingemann and Werner~\cite{Kretschmann2008} gives a quantitative equivalence between the distance between two quantum channels and the distance between their corresponding Stinespring representations, and has since found numerous applications in quantum information theory (see \cite{vomEnde2023} and references therein). More precisely, for two quantum channels $\Phi_1, \Phi_2 : \M{d}\to \M{n}$ with respective Stinespring isometries $V_1, V_2 : \C{d}\to \C{n} \otimes \C{m}$ on a common dilation space $\C{m}$, the following holds true \cite{Kretschmann2008}:
\begin{equation}\label{KSW}
 \min_{U\in \U{m}} \norm{V_1 - (\iden \otimes U) V_2}_{\infty}   \overset{(a)}{\leq}  \sqrt{\norm{\Phi_1 - \Phi_2}_{\diamond}} \overset{(b)}{\leq} \sqrt{2\min_{U\in \U{m}} \norm{V_1 - (\iden \otimes U) V_2}_{\infty}},
\end{equation}
where $\norm{\cdot}_{\infty}$ is the operator norm, $\norm{\cdot}_{\diamond}$ is the diamond norm, and the minimum is over the unitary group $\U{m}\subseteq \M{m}$. Crucially, while the upper bound $(b)$ holds for all $m\in \N$, the lower bound $(a)$ was shown to hold when the common dilation space is large enough: $m\geq 2nd$. Recently, the question about how large the dimension of the dilation space should be for $(a)$ to hold was addressed in \cite{vomEnde2023}. It was shown that $(a)$ holds when $m$ is at least as large as the sum of the Kraus ranks of $\Phi_1$ and $\Phi_2$, but it \emph{cannot} hold for all $m\in \N$, and a modified bound with an extra $\sqrt{2}$ factor was conjectured to hold for all $m\in \N$ \cite{vomEnde2023}:
\begin{equation}\label{eq:KSW-lower-modified}
    \min_{U\in \U{m}} \norm{V_1 - (\iden \otimes U) V_2}_{\infty} \leq \sqrt{2\norm{\Phi_1 - \Phi_2}_{\diamond}}.
\end{equation}
The main result of~\cite{vomEnde2023} proves this bound under the assumption that at least one of the two channels has unit Kraus rank. Moreover, it was shown that the constant $\sqrt{2}$ on the right hand side is optimal. In this paper, we prove the conjecture \eqref{eq:KSW-lower-modified} in full generality.

\begin{theorem}\label{theorem:main}
    For two quantum channels $\Phi_1, \Phi_2 :\M{d}\to \M{n}$ with respective Stinespring isometries $V_1, V_2 : \C{d}\to \C{n} \otimes \C{m}$ on \emph{any} common dilation space $\C{m}$, the following holds true:
\begin{equation}
  \min_{U\in \U{m}} \norm{V_1 - (\iden \otimes U) V_2}_{\infty} \leq   \sqrt{2\norm{\Phi_1 - \Phi_2}_{\diamond}}.
\end{equation} 
\end{theorem}

Our proof builds upon the ideas introduced in \cite{Kretschmann2008, vomEnde2023}, which we recall in Section~\ref{sec:prelim}. 

\section{Preliminaries}\label{sec:prelim}

Let us fix some basic notation first. We denote the algebra of $d_1\times d_2$ complex matrices by $\M{d_1 \times d_2}$, with $\M{d}:= \M{d\times d}$. The unitary group in $\M{d}$ is denoted by $\U{d}$. The identity matrix is denoted by $\iden$. The trace-norm and the operator
norm on matrices are denoted by $\norm{\cdot}_1$ and
$\norm{\cdot}_\infty$, respectively. The minimum eigenvalue of a Hermitian matrix $X$ is denoted by
$\lambda_{\min}(X)$. For Hermitian matrices $X,Y\in \M{d}$, $X\geq Y$ denotes that $X-Y$ is positive
semi-definite. The convex set of quantum states, i.e. positive semi-definite matrices in
$\M{d}$ with unit trace, is denoted by $\cD(\C{d})$. The \emph{fidelity} between two quantum states $\rho, \sigma \in \cD(\C{d})$ is defined as \cite{watrous2018theory} 
\begin{equation}
    f(\rho, \sigma):= \norm{\sqrt{\rho}\sqrt{\sigma}}_1 = \Tr \sqrt{\sqrt{\sigma} \rho \sqrt{\sigma}}.
\end{equation}

A linear map $\Phi:\M{d}\to\M{n}$ is called a quantum channel if it is completely positive
and trace-preserving.
The \emph{diamond} norm of $\Phi$ is defined as \cite{watrous2018theory} 
\begin{equation}
    \norm{\Phi}_{\diamond} := \sup_{\norm{X}_1\leq 1} \norm{(\id \otimes \Phi)(X)}_1 
\end{equation}
where $\id:\M{d}\to \M{d}$ is the identity channel and the supremum is over all $X\in \M{d}\otimes \M{d}$ with $\norm{X}_1\leq 1$. The fidelity between two quantum channels $\Phi_1, \Phi_2:\M{d}\to \M{n}$ is defined as 
\begin{align}
    F(\Phi_1,\Phi_2) &:= \min_{\rho\in \cD(\C{d}\otimes \C{d})  } f\left( (\id \otimes \Phi_1) (\rho), (\id \otimes \Phi_2) (\rho ) \right) \\  
    &= \min_{\ket{\psi}\in \C{d}\otimes \C{d} , \norm{\psi}=1  } f\left( (\id \otimes \Phi_1) (\ketbra{\psi}), (\id \otimes \Phi_2) (\ketbra{\psi} ) \right), \label{eq:channel-fidelity}
\end{align}
where the minimum can be restricted to pure states because of joint concavity of fidelity \cite[Theorem 3.25]{watrous2018theory}. Moreover, the Fuchs-van de Graaf inequalities \cite{Fuchs1999fidelity} show that 
\begin{equation}\label{eq:fuchs-channel}
    1 - F(\Phi_1,\Phi_2) \leq \frac{1}{2}\norm{\Phi_1 - \Phi_2}_{\diamond} \leq \sqrt{1-F^2(\Phi_1,\Phi_2)}.
\end{equation}

Suppose $\Phi_1 , \Phi_2$ admit Stinespring representations \eqref{eq:Stinespring} with isometries $V_1 ,V_2 : \C{d}\to \C{n}\otimes \C{m}$ on a common dilation space $\C{m}$, respectively. Then, the following characterization of channel fidelity $F(\Phi_1, \Phi_2)$ was obtained in \cite[Eq.~(24)]{Kretschmann2008} (see also \cite[Eq.~(3)]{vomEnde2023}):
\begin{align}
    F(\Phi_1,\Phi_2) &= \max_{ \substack{W\in \M{m} \\ \norm{W}_{\infty}\leq 1 }} \min_{\rho\in \cD (\C{d})} \Re \left( \Tr \rho V_1^* (\iden \otimes W) V_2  \right) \label{eq:F-W-1} \\
    &= \frac{1}{2} \max_{ \substack{W\in \M{m} \\ \norm{W}_{\infty}\leq 1 }} \lambda_{\min} \bigl( V_1^* (\iden \otimes W) V_2 + V_2^* (\iden \otimes W^*) V_1 \bigr). \label{eq:F-W-2}
\end{align}

We briefly recall the argument from \cite{Kretschmann2008} to prove \eqref{eq:F-W-1}. Let $\ket{\psi}\in \C{d}\otimes \C{d}$ be a purification of $\rho\in \cD(\C{d})$, so that $\rho=\Tr_{\C{d}}\ketbra{\psi}$. Then, $(\iden\otimes V_i)\ket{\psi}\in \C{d} \otimes \C{n} \otimes \C{m}$ is a purification of $(\id\otimes \Phi_i)(\ketbra{\psi})\in \cD (\C{d} \otimes \C{n})$. By Uhlmann's theorem \cite{Uhlmann1976}, we can write 
\begin{equation}
f\bigl((\id\otimes \Phi_1)(\ketbra{\psi}),(\id\otimes \Phi_2)(\ketbra{\psi})\bigr)
=
\max_{U\in \U{m}}
\left|\Tr \rho\, V_1^*(\iden\otimes U)V_2\right|.
\end{equation}
Since every $\rho\in \cD(\C{d})$ admits such a purification, and since the unit ball of $\M{m}$ is the convex hull of $\U{m}$ and is invariant under multiplication by phases, this shows that \eqref{eq:channel-fidelity} becomes
\begin{equation}
F(\Phi_1,\Phi_2)
=
\min_{\rho\in \cD(\C{d})}
\max_{\substack{W\in \M{m}\\ \norm{W}_\infty\leq 1}}
\Re\left(\Tr \rho\, V_1^*(\iden\otimes W)V_2\right).
\end{equation}
The function being optimized is continuous and affine in both variables, while $\cD(\C{d})$ and the unit ball of $\M{m}$ are compact and convex. Hence, we can swap the $\min$ and $\max$ by using minimax theorem \cite{Sion1958, Simons1995} to obtain \eqref{eq:F-W-1}.

Before proceeding further, we note a simple lemma.

\begin{lemma}\label{lemma:X-Y}
    For any two matrices $X,Y\in \M{d_1 \times d_2}$, $(X-Y)^*(X-Y) \leq 2(X^*X + Y^*Y)$.
\end{lemma}
\begin{proof}
    Since $(X+Y)^*(X+Y)=X^*X+Y^*Y + X^*Y+Y^*X \geq 0$, we get 
    \begin{align}
        (X-Y)^*(X-Y) &= X^*X+Y^*Y - X^*Y - Y^*X \leq 2(X^*X+Y^*Y).
    \end{align}
\end{proof}

The heart of the proof of Theorem~\ref{theorem:main} is contained in the following lemma, which was already identified in \cite[Section 3]{vomEnde2023} as one possible route towards proving Theorem~\ref{theorem:main}.

\begin{lemma}\label{lemma:main}
    Let $V_1,V_2 : \C{d}\to \C{n}\otimes \C{m}$ be linear isometries. Define the map
    \begin{align}
        \Gamma : \M{m} &\to \mathbb{R} \\
                  W &\mapsto \lambda_{\min}\bigl( V_1^* (\iden \otimes W) V_2 + V_2^* (\iden \otimes W^*) V_1 \bigr).
    \end{align}
    Then, the following holds true\footnote{The continuity of $\Gamma:\M{m}\to \mathbb{R}$ follows easily from the variational characterization of eigenvalues for Hermitian matrices \cite[Chapter 3]{Bhatia1997matrix}. Hence, supremum of $\Gamma$ over any compact subset in $\M{m}$ is attained.}:
    \begin{equation}
        2\max_{ \substack{W\in \M{m} \\ \norm{W}_{\infty}\leq 1 }} \Gamma(W)  \leq 2 + \max_{U\in \U{m}} \Gamma (U).
    \end{equation}
\end{lemma}
\begin{proof}
Let us first note that for unitary matrices $U\in \U{m}$, 
\begin{align}
    (V_1 - (\iden \otimes U)V_2 )^* (V_1 - (\iden \otimes U)V_2 ) = 2 \iden - V_1^* (\iden \otimes U) V_2 - V_2^* (\iden \otimes U^*) V_1,
\end{align}
so that $\norm{V_1 - (\iden \otimes U)V_2}_{\infty}^2 = 2 -\Gamma(U)$. 

Now, fix a contraction $W\in \M{m}$, $\norm{W}_{\infty}\leq 1$. We can use polar decomposition to write $W=UP$ for some unitary $U\in \U{m}$ and $P\in \M{m}$ satisfying $\mathbb{0} \leq P \leq \iden$, so that $W^*W=P^2$ and $U-W=U(\iden - P)$. Note that
\begin{equation}
    V_1 - (\iden \otimes U)V_2 = \underbrace{V_1 - (\iden \otimes W)V_2}_{:= X} - \underbrace{ \bigl(\iden \otimes (U-W) \bigr) V_2}_{:=Y},
\end{equation}
where $Y= \bigl(\iden \otimes (U-W) \bigr) V_2 = \bigl(\iden \otimes U(\iden - P) \bigr) V_2$ satisfies 
\begin{align}\label{eq:Y*Y<=Z*Z}
    Y^*Y = V_2^* \Bigl(\iden \otimes (\iden - P)^2 \Bigr) V_2 \leq V_2^* \Bigl(\iden \otimes \bigl(\iden - P^2 \bigr) \Bigr) V_2 =: Z^*Z,
\end{align}
where $(\iden - P)^2 \leq \iden - P^2$ holds because $\mathbb{0}\leq P\leq \iden$, and $Z:= \bigl(\iden \otimes \sqrt{\iden - P^2} \bigr) V_2$. Hence,
\begin{align}
    (X-Y)^* &(X-Y) \nonumber \\ &\leq 2(X^*X + Y^*Y)  \\
    &\leq 2(X^*X + Z^*Z) \\
    &= 2 \Bigl( \iden + V_2^* \bigl(\iden \otimes P^2 \bigr)V_2 - V_1^* (\iden \otimes W) V_2 - V_2^*(\iden \otimes W^*)V_1 + V_2^* \bigl(\iden \otimes (\iden-P^2) \bigr) V_2  \Bigr) \\
    &= 2 \bigl( 2\iden - V_1^* (\iden \otimes W) V_2 - V_2^*(\iden \otimes W^*)V_1 \bigr),
\end{align}
where the first inequality follows from Lemma~\ref{lemma:X-Y} and the second one follows from \eqref{eq:Y*Y<=Z*Z}. Thus, for every contraction $W\in \M{m}$, $\norm{W}_{\infty}\leq 1$, taking $U\in \U{m}$ to be the polar unitary shows
\begin{align}
    2-\Gamma(U) = \norm{V_1 - (\iden \otimes U)V_2}_{\infty}^2 = \norm{X-Y}_{\infty}^2 \leq 2(2- \Gamma(W)).
\end{align}
Rearranging the terms yields $2\Gamma(W)\leq 2 + \Gamma(U)$, which finishes the proof.
\end{proof}

With all the setup in place, we can now prove Theorem~\ref{theorem:main}.

\section{Proof of Theorem~\ref{theorem:main}}\label{sec:proof}

Fix Stinespring isometries $V_1, V_2 : \C{d}\to \C{n}\otimes \C{m}$ for the two channels, and let $\Gamma : \M{m}\to \mathbb{R}$ be defined as in Lemma~\ref{lemma:main}. Recall from \eqref{eq:F-W-2} that
\begin{equation}
    2F(\Phi_1 , \Phi_2) = \max_{ \substack{W\in \M{m} \\ \norm{W}_{\infty}\leq 1 }}  \Gamma(W) = \Gamma (W_\star),
\end{equation}
where $W_{\star}\in \M{m}$ is the contraction that attains the maximum. Then, Lemma~\ref{lemma:main} shows that there exists a unitary $U_\star\in \U{m}$ such that $2\Gamma(W_\star) \leq 2 + \Gamma(U_\star)= 4 - \norm{V_1 - (\iden \otimes U_\star)V_2}_{\infty}^2$. Hence, using Fuchs-van de Graaf inequality \eqref{eq:fuchs-channel}, we obtain
\begin{align}
4 \left( 1 - \frac{1}{2} \norm{\Phi_1 - \Phi_2}_{\diamond} \right) \leq 4 F(\Phi_1, \Phi_2) =2\Gamma(W_\star) \leq   4 - \norm{V_1 - (\iden \otimes U_\star)V_2}_{\infty}^2,
\end{align}
which, upon rearranging terms, yields the desired claim:
\begin{equation*}
\pushQED{\qed} 
\norm{V_1 - (\iden \otimes U_\star) V_2}_{\infty} \leq   \sqrt{2\norm{\Phi_1 - \Phi_2}_{\diamond}}.\qedhere
\popQED
\end{equation*}

For completeness, we recall the example from \cite{vomEnde2023} to show that $\sqrt{2}$ in Theorem~\ref{theorem:main} is optimal. Consider unitary channels $\Phi_1,\Phi_2:\M{n}\to\M{n}$, $\Phi_k (\cdot) = V_k (\cdot)V_k^*$, defined by $V_1=\iden_n$ and $V_2=\operatorname{diag}(e^{2\pi ij/n})_{j=1}^n$, $n\geq 2$. Here, the dilation space is $\mathbb{C}$, and a simple geometric argument shows
\begin{equation}
\min_{z\in \U{1}}\norm{V_1-zV_2}_\infty
=
2\sin\left(\frac{\pi(n-1)}{2n}\right).
\end{equation}
Moreover, since the convex hull of the spectrum of $V_1^*V_2=V_2$ contains the origin, the two channels are perfectly distinguishable \cite{Acin2001}, so that $\norm{\Phi_1-\Phi_2}_\diamond=2$. Therefore,
\begin{equation}
\frac{
\min_{z\in \U{1}} \norm{V_1-zV_2}_\infty
}{
\sqrt{\norm{\Phi_1-\Phi_2}_\diamond}
}
=
\sqrt{2}\sin\left(\frac{\pi(n-1)}{2n}\right)
\longrightarrow \sqrt{2} \quad \text{as } n\to \infty,
\end{equation}
so no dimension-independent constant smaller than $\sqrt{2}$ can hold in Theorem~\ref{theorem:main}.

\section*{Acknowledgements}
I acknowledge support from the Deutsche Forschungsgemeinschaft (DFG, German Research Foundation) via TRR 352 – Project-ID 470903074. 

\bibliographystyle{alpha}
\bibliography{references}
\vspace{0.25cm}
\hrule

\end{document}